\documentclass[conference]{IEEEtran}

\IEEEoverridecommandlockouts

\usepackage[cmex10]{amsmath}
\usepackage{array}
\usepackage{mdwmath}
\usepackage{mdwtab}
\usepackage{eqparbox}
\usepackage{amsmath}
\usepackage{amssymb,amsbsy}
\usepackage{amsthm}
\usepackage{graphicx}
\usepackage{color}

\newcommand{\Fig}[1]{Fig.~\ref{#1}}
\graphicspath{{figures/}}

\newtheorem{lemma}{Lemma}
\newtheorem{theorem}{Theorem}
\newtheorem{corollary}{Corollary}
\newtheorem{remark}{Remark}
\newtheorem{example}{Example}

\newcommand{\F}{\mathbb{F}}

\begin{document}

\sloppy

\title{Upper Bounds on the Minimum Distance of Quasi-Cyclic LDPC codes Revisited}

\author{
  \IEEEauthorblockN{Alexey Frolov and Pavel Rybin}
	
  \IEEEauthorblockA{\small Inst. for Information Transmission Problems\\
    Russian Academy of Sciences\\Moscow, Russia\\
    Email: \{alexey.frolov, prybin\}@iitp.ru
  }
}




\maketitle
\begin{abstract}
Two upper bounds on the minimum distance of \mbox{type-$1$} quasi-cyclic low-density parity-check (QC LDPC) codes are derived. The necessary condition is given for the minimum code distance of such codes to grow linearly with the code length.
\end{abstract}

\section{Introduction}
In this paper we investigate the minimum code distance of QC LDPC codes \cite{Ta1999, MD, F}. These codes form an important subclass of LDPC codes \cite{G, Ta1981}. These codes also are a subclass of protograph-based LDPC codes \cite{Th}. QC LDPC codes can be easily stored as their parity-check matrices can be easily described. Besides such codes have efficient encoding \cite{CZLF} and decoding \cite{KFL} algorithms. All of these makes the codes very popular in practical applications. 

In \cite{MD} an upper bound on the minimum distance of QC LDPC codes is derived for the case when the base matrix has all the elements equal to one. In this case the minimum code distance is upper bounded by a quantity $(m+1)!$, where $m$ is a height of a base matrix and at the same time (due to the structure of the base matrix) the number of ones in a column of the base matrix. In \cite{SV} the results of \cite{MD} are generalized for the case of type-$w$ QC LDPC codes (see Theorems 7 and 8 in \cite{SV}). Unfortunately these estimates can be applied only to a certain parity-check matrix. In this paper we obtain the upper bounds which are valid for any code from the ensemble of QC LDPC codes with the given degree distribution. This allows us to formulate the necessary condition for the minimum code distance of such codes to grow linearly with the code length. We consider only the case of so-called \mbox{type-$1$} QC LDPC codes. 

Our contribution is as follows. Two upper bounds on the minimum distance of \mbox{type-$1$} quasi-cyclic low-density parity-check (QC LDPC) codes are derived. The necessary condition is given for the minimum code distance of such codes to grow linearly with the code length.

The structure of the paper is as follows. In section~\ref{sect2} the preliminaries on QC LDPC codes are given. In section~\ref{sect3} the bounds are derived and analyzed.

\section{Preliminaries}\label{sect2}
In this paper we only consider binary codes. Let $w$ be some positive integer. Consider a matrix of size $m \times n$
\[
\mathbf{H}^{(W)} = \left[ h_{i,j}\right] \in \{0,1, \ldots, w\} ^ {m \times n}.
\]
In what follows the matrix will be referred to as the weight matrix\footnote{in the literature the matrix is called a base matrix or a proto-matrix.}.

Let us construct a parity-check matrix $\mathbf{H}$ of the QC LDPC code $\mathcal C$. For this purpose we extend the matrix $\mathbf{H}^{(W)}$ with circulant matrices (circulants) as follows: 
\begin{equation*}
\mathbf{H} = 
\left[
\begin{array}{cccc}
\mathbf{P}_{1,1}  & \mathbf{P}_{1,2}  & \cdots & \mathbf{P}_{2,n}      \\
\mathbf{P}_{2,1}  & \mathbf{P}_{2,1}  & \cdots & \mathbf{P}_{2,n}      \\
\vdots                 & \vdots                 & \ddots & \vdots       \\
\mathbf{P}_{m,1} & \mathbf{P}_{m,1} & \cdots & \mathbf{P}_{m,n}   
\end{array}
\right] \in \F_2^{ms \times ns},
\end{equation*}
where $\mathbf{P}_{i,j}$ is a circulant over a binary field $\F_2$ of size $s \times s$ ($s \geq w$) and of weight\footnote{the weight of a circulant is a weight of its first row.} $h_{i,j}$, $i = 1, \ldots, m$; $j = 1, \ldots, n$.

Let us denote the length of the code $\mathcal C$ by $N = ns$, such inequality follows for the rate of the code
$$
R(\mathcal C) \geq1-\frac{m}{n}.
$$

\begin{remark}
It is easy to see that the obtained code is in fact quasi-cyclic. Consider the codeword ${\bf c} \in \mathcal C$. Let us split the codeword into $n$ subblocks in accordance to the structure of the parity-check matrix $\mathbf{H}$\textup:
\[
{\bf c} = \left( {\bf c}_1, {\bf c}_2, \ldots, {\bf c}_n \right),
\] 
then note that if we apply the same cyclic shifts in each subblock we again obtain a codeword of $\mathcal C$.
\end{remark}

\begin{remark}
The constructed code is a \mbox{type-$w$} QC LDPC code. In what follows we will only consider \mbox{type-$1$} QC LDPC codes, i.e. $w = 1$. In this case the matrix $\mathbf{H}^{(W)}$ can be considered as a matrix over $\F_2$.
\end{remark}

Let $\F$ be some field, by $\F[x]$ we denote the ring of all the polynomials with coefficients in $\F$. It is well-known that the ring of circulants of size $s \times s$ over $\F$ is isomorphic to the factor ring $\F^{(s)}[x] = \F[x]/\left( x^s-1 \right) $. Thus with the parity-check matrix $\mathbf{H}$ we associate a polynomial parity-check matrix $\mathbf{H}(x) \in \left( \F_2^{( s )}[x] \right)^{m \times n}$:
\begin{equation*}
\mathbf{H}(x) = 
\left[
\begin{array}{cccc}
p_{1,1}(x)  & p_{1,2}(x)  & \cdots & p_{1,n}(x)      \\
p_{2,1}(x)  & p_{2,1}(x)  & \cdots & p_{2,n}(x)      \\
\vdots      & \vdots      & \ddots & \vdots          \\
p_{m,1}(x)  & p_{m,1}(x)  & \cdots & p_{m,n}(x)  
\end{array}
\right],
\end{equation*}
where $p_{i,j}(x) = \sum\nolimits_{t=1}^{s} {P_{i,j}(t,1)x^{t-1}}$, by $P_{i,j}(t,1)$ we mean an element at the intersection of the $t$-th row and the first column in the matrix $P_{i,j}$. 

\begin{example}
Matrices
\[
\mathbf{H}^{(W)} = 
\left[
\begin{array}{ccc}
0 & 1 & 1    \\
1 & 0 & 1   \\
\end{array}
\right]
\]
и
\[
\mathbf{H}(x) = 
\left[
\begin{array}{ccc}
0  & x^2  & x    \\
1 & 0 & x^2   \\
\end{array}
\right].
\]
correspond to the parity-check matrix
\[
\mathbf{H} = 
\left[
\begin{array}{c|c|c}
\begin{array}{ccc} 0 & 0 & 0 \\ 0 & 0 & 0  \\ 0 & 0 & 0 \end{array}  & \begin{array}{ccc} 0 & 1 & 0 \\ 0 & 0 & 1  \\ 1 & 0 & 0 \end{array}  & \begin{array}{ccc} 0 & 0 & 1 \\ 1 & 0 & 0  \\ 0 & 1 & 0 \end{array}     \\
\noalign{\vskip 0.1cm} 
\hline 
\noalign{\vskip 0.1cm} 
\begin{array}{ccc} 1 & 0 & 0 \\ 0 & 1 & 0  \\ 0 & 0 & 1 \end{array} & \begin{array}{ccc} 0 & 0 & 0 \\ 0 & 0 & 0  \\ 0 & 0 & 0 \end{array}  & \begin{array}{ccc} 0 & 1 & 0 \\ 0 & 0 & 1  \\ 1 & 0 & 0 \end{array}    \\
\end{array}
\right]
\]
\end{example}

\begin{remark}[Connection to protograph-based LDPC codes]
QC LDPC codes is a subclass of protograph-based LDPC codes. In this case $\mathbf{H}^{(W)}$ is the adjacency matrix of a protograph and permutation matrices can only be chosen from circulants.
\end{remark}

Let us associate the vector
$$
{\bf c} = ( {\bf c}_{1}, {\bf c}_{2}, \ldots, {\bf c}_{n} ),
$$ 
where 
$$
{\bf c}_{i} = (c_{i,1}, c_{i,2}, \ldots, c_{i,s}), \quad i = 1,2, \ldots, n,
$$
to the vector of polynomials
$$
{\bf c}(x) = ( c_{1}(x), c_{2}(x), \ldots, c_{n}(x) ),
$$
where $c_{i}(x) = \sum\nolimits_{t=1}^{s} {c_{i,t} x^{t-1}}$.

Clear, that
\[
{\bf H} {\bf c}^T = {\bf 0} \quad (\text{in the filed} \:\:\: \F_2) 
\]
is equivalent to
\[
{\bf H}(x) {\bf c}^T(x) = {\bf 0} \quad (\text{in the ring} \:\:\: \F^{(s)}_2[x]).
\]

By the weight of polynomial $f(x)$ we mean the number of non-zero coefficients. We denote the weight by $||f(x)||$. Let us define the weight of the vector of polynomials ${\bf c}(x) = ( c_{1}(x), c_{2}(x), \ldots, c_{n}(x) )$ as follows
$$
||{\bf c}(x)|| = \sum\limits_{i=1}^n {||c_i(x)||}.
$$

\section{Minimum code distance}\label{sect3}

Let us denote the minimum code distance of the code $\mathcal C$ by $D(\mathcal C)$. First we derive a simple bound.
\begin{theorem}
Let $\mathcal C$ be a \mbox{type-$1$} QC LDPC code with the weight matrix $\mathbf{H}^{(W)}$ and let $d$ be a minimal code distance of the code which corresponds to the parity-check matrix $\mathbf{H}^{(W)}$, then
\begin{equation}\label{est_1}
D(\mathcal C) \leq ds.
\end{equation}
\end{theorem}
\begin{proof}
Let $c_W$ be a codeword of weight $d$ of the code with the parity-check matrix $\mathbf{H}^{(W)}$, $S = \operatorname{supp}(c_W)$ and let $ f(x) = \sum\nolimits_{j=0}^{s-1}{x^j}$. Let us construct a codeword ${\bf c}(x) \in \mathcal C$. For $i = 1, \ldots, n$
\[
c_i (x)  =  \left\{ {\begin{array}{rc}
  f(x), & \quad i \in S, \\ 
  0,      & \quad \text{otherwise.} 
\end{array}} \right.
\]
We only need to note that $x^j  f(x)= f(x)  \: \forall j = 0, \ldots, s-1$,
hence
$$
{\bf H}(x) {\bf c}^T(x) = f(x) \mathbf{H}^{(W)} c_W^T = {\bf 0}.
$$
\end{proof}

Let us introduce the notation of a submatrix. Let $\bf A$ be some matrix of size $M \times N$. Let $I \subseteq \{1,2,\ldots,M\}$ be a subset of rows, $J \subseteq \{1,2,\ldots,N\}$ -- subset of columns. By $\mathbf{A}_{I,J}(x)$ we denote a submatrix of $\mathbf{A}$ which contains only rows with numbers in $I$ and only columns with numbers in $J$. If $I = \{1,2,\ldots,M\}$, then we use a notation $\mathbf{A}_{J}(x)$.

To derive the second estimate we start with the Lemma given in \cite{SV}. This lemma is the generalization of Theorem~2 from \cite{MD} and shows how to construct codewords of QC LDPC codes. As in this paper we work with \mbox{type-$1$} QC LDPC codes we formulate the Lemma for such codes only.

\begin{lemma}[Smarandache and Vontobel, \cite{SV}]\label{CW}
Let $\mathcal C$ be a \mbox{type-$1$} QC LDPC code with the polynomial matrix ${\bf H}(x)$. Let $J  \subset \{1, 2, \ldots, n\}$, $|J| = m+1$ and let $\Delta_j(x) = \det \left( \mathbf{H_{J \backslash \{j\} }} (x) \right)$, then a word ${\bf c}(x) = ( c_{1}(x), c_{2}(x), \ldots, c_{n}(x) )$, where
\[
c_j (x)  =  \left\{ {\begin{array}{rc}
  {\Delta_j}(x), & \quad j \in J, \\ 
  0,      & \quad \text{otherwise.} 
\end{array}} \right.
\]
is a codeword of $\mathcal C$.
\end{lemma}
\begin{proof}
Let us show that ${\bf s}(x) = {\bf H}(x) {\bf c}^T(x) = {\bf 0}$ in the ring $\F_2^{(s)}[x]$. We only give the proof for the first element of the syndrome:
\[
{s_1}(x) = \sum\limits_{j = 1}^n {{p_{1,j}}(x){c_j}(x) = \sum\limits_{j \in J} {{p_{1,j}}(x){\Delta _j}(x)} }.
\]

Let $J = \{j_1, j_2, \ldots, j_{m+1} \}$. Note, that
\begin{eqnarray*}
{s_1}(x)  &=& \det \left[ {\begin{array}{cccc}
  {{p_{1,{j_1}}}}(x)&{{p_{1,{j_2}}}}(x)& \cdots &{{p_{1,{j_{m + 1}}}}}(x) \\ 
  \hline
  {{p_{1,{j_1}}}}(x)&{{p_{1,{j_2}}}}(x)& \cdots &{{p_{1,{j_{m + 1}}}}}(x) \\ 
  {{p_{2,{j_1}}}}(x)&{{p_{2,{j_2}}}}(x)& \cdots &{{p_{2,{j_{m + 1}}}}}(x) \\ 
   \vdots & \vdots & \ddots & \vdots  \\ 
  {{p_{m,{j_1}}}}(x)&{{p_{m,{j_2}}}}(x)& \cdots &{{p_{m,{j_{m + 1}}}}}(x)
\end{array}} \right]\\
&=& 0,
\end{eqnarray*}
as the matrix contains two identical rows. Analogously one can carry out the proof for the rest elements of the syndrome.
\end{proof}

We need to introduce a notation $\overline{l}(t_1, t_2)$. Let us arrange the columns of the matrix ${\bf H}^{(W)}$  in ascending order of their weights (i.e. the first columns are of small weight, the last ones are of large weight). In what follows we assume the columns of the matrix ${\bf H}^{(W)}$ to be in this order.  Let $l_j$ be a weight of the $j$-th column of  ${\bf H}^{(W)}$, $t_2 > t_1$,  then
\[
\overline{l}(t_1, t_2) = \frac{1}{t_2-t_1+1}\sum\limits_{i=t_1}^{t_2} l_i.
\] 
Let $l_\text{max}$ and $l_\text{min}$ be accordingly maximum and minimum column weights in ${\bf H}^{(W)}$ (in our case $l_\text{max} = l_n$ and $l_\text{min} = l_1$). 

\begin{example}
Let us consider the matrix ${\bf H}^{(W)}$ with $n = 48$. Let $\Lambda(x) = 12 x^2 + 24 x^3 + 12 x^4$ be the variable degree \textup(column weight\textup) distribution polynomial for ${\bf H}^{(W)}$ . Recall that
\[
\Lambda(x) = \sum\limits_{i=1}^{l_\text{max}} \Lambda_i x^i,
\]
where $\Lambda_i$ is a number of columns of weight $i$. For more info on degree distribution polynomials see \textup{\cite{UR}}.

The dependency $\overline{l}(2, t)$ for this case is shown in \Fig{av_weight}.

\begin{figure}[t]
\centering
\includegraphics[width=0.48\textwidth]{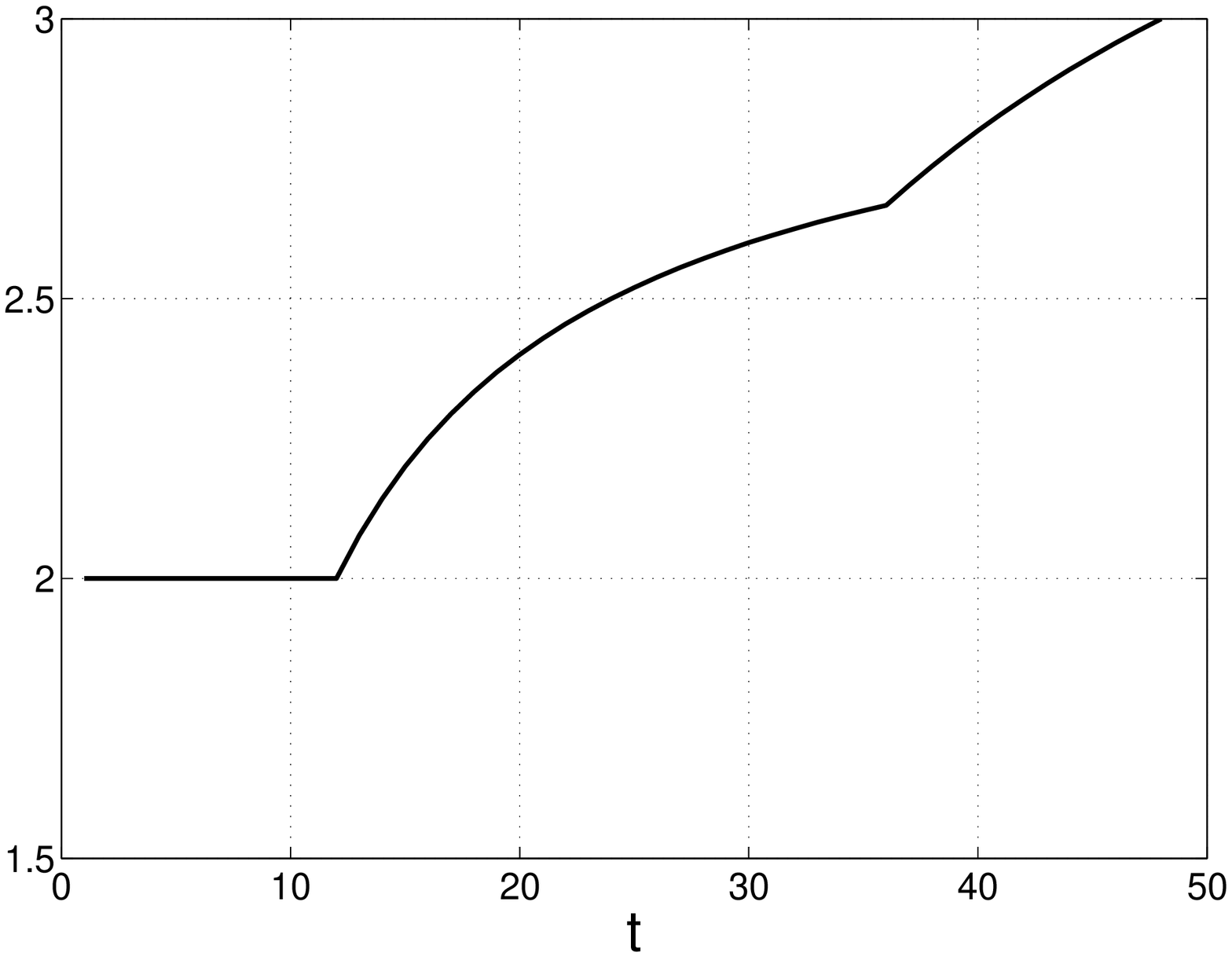}
\caption{The dependency $\overline{l}(2, t)$}
\label{av_weight}
\end{figure}
\end{example}

\begin{theorem}\label{theorem2}
Let $\mathcal C$ be a \mbox{type-$1$} QC LDPC code with the weight matrix ${\bf H}^{(W)}$ of size $m \times n$, let $k$, $1 \leq k \leq m$, be the largest integer for which
$l_{m+2-k} \geq k$ and let $\ell =  \overline{l}(2, m+1-k)$ then
\begin{equation}\label{est_2}
D(\mathcal C) \leq (m+1) k!  \ell^{m-k}.
\end{equation}
\end{theorem}
\begin{proof}
Recall that the columns of the matrix ${\bf H}^{(W)}$  are in ascending order of their weights. Let $J = \{1, 2, \ldots, m+1\}$. Let us construct a codeword ${\bf c}(x)$ in accordance to Lemma~\ref{CW}. The last $n-|J|$ positions ${\bf c}(x)$ are equal to zero. 

Consider $\Delta_1(x)$. Note, that 
$$
||\Delta_1(x)|| \leq \prod\limits_{i = 1}^{m} {\min \{i, l _{m+2-i}}\} = k! \prod\limits_{j = 2}^{m + 1 - k} {{l _j}},
$$ 
where $l _j$ is a weight of the $j$-th column in ${\bf H}^{(W)}_J$. This inequality follows from the fact that the sum for $\Delta_1(x)$ contains at most $k! \prod\nolimits_{j = 2}^{m +1 - k} {{l _j}}$ terms. Each of this terms is a monomial. Since
$$
\prod\limits_{j = 2}^{m + 1 - k} {{l _j}} \leq \ell^{m-k},
$$
then
$$
||\Delta_1(x)|| \leq k!  \ell^{m-k}.
$$
Similar inequalities hold for all the $\Delta_j(x)$, $j \in J$.  As there are at most $m+1$ non-zero positions in a codeword ${\bf c}(x)$, then 
\[
||{\bf c}(x)|| \leq (m+1) k!  \ell^{m-k}.
\]

We should also consider the case when all $\Delta_j(x) = 0 \quad \forall j \in J$. In this case Lemma~\ref{CW} gives a zero codeword. We proceed as follows. We find a non-zero minor of the maximal order $r$, $r < m$ in the matrix $\mathbf{H}_J(x)$. Let $I$ be a set of row numbers, $S$ be a set of column numbers, such that $\mathbf{H}_{I,S}(x)$ is the minor. Let $S' = S \cup j$, $j \in J \backslash S$. Consider the submatrix $\mathbf{H}_{I,S'}(x)$. We construct a codeword for this submatrix in accordance to Lemma~\ref{CW}. Note, that this word contains al least one non-zero position. After appending this word with zeros on positions  $\{1,2\ldots,n\} \backslash S'$, we obtain a codeword for the matrix $\mathbf{H}(x)$, as all the minors of bigger order are equal to zero. In this case we have 
$$
D(\mathcal C) \leq (r+1) k!  \ell^{m-k} < (m+1) k!  \ell^{m-k},
$$
this completes the proof.
\end{proof}

\begin{remark}
Note that the bound is better for regular codes \textup(see \Fig{av_weight}\textup). In this case we have \textup(let $\ell$ be the column weight, it is easy to check, that $k = \ell$\textup)
\[
D(\mathcal C) \leq (m+1) \ell!  \ell^{m-\ell}.
\]
If the base matrix is the all one matrix \textup($\ell = m$\textup), we obtain the bound from \textup{\cite{MD}}.
\end{remark}

\begin{remark}
Note that the estimate \textup{(\ref{est_2})} does not depend on $s$. If $m$ and $n$ are fixed and $s \to \infty$, then in accordance to the estimate \textup{(\ref{est_2})} $D(\mathcal C)$ is upper bounded by a constant. We also note, that in \textup{\cite{ZZ}} it is proved that there exist protograph-based LDPC codes with the following properties: the minimum distance of such codes grows linearly with the code length while the sizes of the base matrix \textup($m$ and $n$\textup) are fixed.
\end{remark}

\begin{corollary}
Thus, for the minimum code distance $D(\mathcal C)$ to grow linearly with the code length $N=ns$ it is necessary, that the estimates \textup{(\ref{est_1})} and \textup{(\ref{est_2})} grow linearly with $N$. 
\end{corollary}

\section{Conclusion}
Two upper bounds on the minimum distance of \mbox{type-$1$} quasi-cyclic low-density parity-check (QC LDPC) codes are derived. The necessary condition is given for the minimum code distance of such codes to grow linearly with the code length.

\section*{Acknowledgment}
The authors thank V.V.~Zyablov for the numerous advice and recommendations.

\end{document}